\documentclass[12pt]{article}
\usepackage{graphicx}
\usepackage{amssymb,amsmath,amsfonts,amsthm}
\usepackage{pdflscape,array}
\usepackage[T1]{fontenc}
\usepackage[utf8]{inputenc}

\voffset = - 1.2cm
\input xy
\xyoption{all}

\theoremstyle{plain}
\newtheorem{theorem}{Theorem}
\newtheorem{corollary}[theorem]{Corollary}
\newtheorem{proposition}[theorem]{Proposition}
\newtheorem{lemma}[theorem]{Lemma}

\theoremstyle{definition}
\newtheorem{definition}[theorem]{Definition}

\newtheorem{example}[theorem]{Example}

\numberwithin{equation}{section}
\numberwithin{theorem}{section}

\begin{document}

\centerline{\Large {\bf Geometric features of Vessiot--Guldberg Lie algebras}}
\vskip 0.15cm
\centerline{\Large{\bf of conformal and Killing  vector fields on $\mathbb{R}^2$}}
\vskip 0.35cm

\centerline{M.M. Lewandowski and J. de Lucas}
\vskip 0.7cm

\centerline{Department of Mathematical Methods in Physics, University of Warsaw,}
\centerline{ul. Pasteura 5, 02-093, Warszawa, Poland}
\medskip

\begin{abstract}
This paper locally classifies finite-dimensional Lie algebras of conformal and Killing vector fields on $\mathbb{R}^2$ relative to an arbitrary pseudo-Riemannian metric. Several results about their geometric properties are detailed, e.g. their invariant distributions and induced symplectic structures. Findings are illustrated with two examples of physical nature: the Milne--Pinney equation and the projective Schr\"odinger equation on the Riemann sphere.
\end{abstract}



\section{Introduction}
The so-called {\it infinitesimal groups of transformations} were introduced by Sophus Lie towards the end of the XIX century so as to study differential equations \cite{HA75}. Nowadays such structures are referred to as {\it Lie algebras of vector fields}, and they play a key role in the research on differential equations \cite{Olver}. 

The local classification of the finite-dimensional real Lie algebras of vector fields on  the plane was accomplished by Lie \cite{HA75,Li80}. Gonz\'alez-L\'opez, Kamran, and Olver retrieved his classification via modern differential geometric techniques while solving unclear points in Lie's work that had been  misunderstood in the previous literature \cite{GKP92}. We hereupon call their classification the {\it GKO classification}.

Our article focuses upon finite-dimensional Lie algebras of conformal and Killing vector fields on $\mathbb{R}^2$ relative to a pseudo-Riemannian metric. A {\it conformal vector field} relative to a pseudo-Riemannian metric $g$ on a manifold $M$ is a vector field $X$ satisfying that $\mathcal{L}_Xg=fg$ for a function $f$ on $M$. If $f=0$, then $X$ is called a {\it Killing vector field} relative to $g$. Lie algebras of conformal and Killing vector fields are relevant due to their applications to Einstein equations \cite{SW72}, covariant quantizations \cite{DLO99,LO99}, and differential equations \cite{Ru16,Ru16II,HLT17,LA08}. 

The problem of classifying Lie algebras of conformal and/or Killing vector fields on types of manifolds has drawn certain attention \cite{KHBK15}. For instance, the local form of Lie algebras of conformal vector fields relative to flat pseudo-Riemannian metrics on a manifold $M$ is known. The case $\dim M=2$ is the most puzzling one, as it leads to an infinite-dimensional Lie algebra of conformal vector fields \cite{BL00,Ta49}. 
To this respect, Boniver and Lecomte proved that the Lie algebras of conformal polynomial vector fields on $\mathbb{R}^2$ relative to $\eta_{\pm}:={\rm d}x\otimes {\rm d}x\pm {\rm d}y\otimes {\rm d}y$ are maximal in the Lie algebra of polynomial vector fields in the variables $x,y$ \cite{BL00}. It is also interesting to study which finite-dimensional Lie algebras of conformal or Killing vector fields determine second-order ordinary differential equations \cite{Ru16,Ru16II}.

The local structure of finite-dimensional Lie algebras of conformal vector fields relative to a flat pseudo-Riemannian metric on $\mathbb{R}^2$ was studied in \cite{GL16}. This result is here extended to finite-dimensional Lie algebras of conformal vector fields relative to any pseudo-Riemannian metric on $\mathbb{R}^2$ by using their conformal flatness \cite{GL16,Ta49}. Our work also performs a local classification of Lie algebras of Killing vector fields on the plane relative to an arbitrary pseudo-Riemannian metric. Simple arguments are given so as to classify the so-called invariant distributions of finite-dimensional Lie algebras of vector fields on $\mathbb{R}^2$, which much simplifies the straightforward but long approach proposed in \cite{GL16}. This result is interesting as it appears in the analysis of finite-dimensional Lie algebras of conformal vector fields on the plane \cite{GKP92,GL16}.

More specifically, we here prove that all conformal Lie algebras of vector fields relative to an arbitrary pseudo-Riemannian metric on $\mathbb{R}^2$ are, up to a local diffeomorphism, the Lie subalgebras of the Lie algebras I$_{7}$ and I$_{11}$ of the GKO classification. Meanwhile, the Lie algebras of Killing vector fields on $\mathbb{R}^2$ are locally diffeomorphic to the Lie subalgebras of the Lie algebras I$_4$, P$^{\alpha=0}_1$, P$_2$, and P$_3$ of the GKO classification. The pseudo-Riemannian metrics associated with these Lie algebras are constructed via a certain type of tensor fields, the so-called {\it Casimir tensor fields} \cite{BHLS15}, derived by means of quadratic {\it Casimir elements} of the above-mentioned Lie algebras \cite{SW14}. This result represents a new application of the theory of Casimir tensor fields initiated in \cite{BHLS15}. Our classifications are detailed in Table \ref{table3}. 

Finally, our findings are applied to  Milne--Pinney equations and projective $t$-dependent Schr\"odinger equations, which are relevant differential equations frequently occurring in physics. These are types of {\it Lie--Hamilton systems} \cite{CLS12}, namely they are differential equations describing the integral curves of a $t$-dependent vector field taking values in a finite-dimensional Lie algebra of Hamiltonian vector fields relative to a {\it Poisson bivector}. The Poisson bivectors associated with above-mentioned differential equations were obtained in previous works \cite{BBHLS15,BHLS15} by means of tedious calculations or {\it ad hoc} considerations. In this work, it is shown that they can be derived geometrically in an easy manner by our here developed application of  {\it Casimir tensor fields}.

The structure of the paper goes as follows. Section 2 addresses an introduction to finite-dimensional Lie algebras of vector fields. Section 3 surveys the theory of conformal and Killing Lie algebras of vector fields on $\mathbb{R}^2$. Section 4 is devoted to the classification of finite-dimensional Lie algebras of conformal and Killing vector fields on the plane, respectively. Section 5 addresses the calculation of invariant distributions for finite-dimensional Lie algebras of vector fields. Finally, Section 6 illustrates several applications of our results to Milne--Pinney and projective Schr\"odinger equations on the complex projective space $\mathbb{C}P^1$. 

\section{Vessiot--Guldberg Lie algebras}\label{ALVG}
This section surveys known results on the theory of Lie algebras of vector fields on the plane. Special attention is paid to finite-dimensional Lie algebras of vector fields, the so-called {\it Vessiot--Guldberg Lie algebras} \cite{JdL11}. To simplify our presentation,  manifolds are hereafter assumed to be connected.

Let $V$ be a Lie algebra with a Lie bracket $[\cdot,\cdot]:V\times V\rightarrow V$. If $\mathcal{A}$ and $\mathcal{B}$ are subsets of $V$, then  $[\mathcal{A},\mathcal{B}]$ is defined to be the linear subspace of $V$ generated by the Lie brackets between elements of $\mathcal{A}$ and $\mathcal{B}$.

A {\it Stefan-Sussmann distribution} on $M$ is a subset  $\mathcal{D}\subset TM$ such that $\mathcal{D}_\xi:=T_\xi M\cap \mathcal{D}$ is not empty for every $\xi\in M$. To simplify our terminology, we will refer to Stefan-Sussmann distributions as distributions. The dimension of $\mathcal{D}_\xi$ is called the {\it rank} of $\mathcal{D}$ at $\xi$. The distribution $\mathcal{D}$ is {\it regular} at $\xi\in M$ if the rank of $\mathcal{D}$ is constant at points of an open $U\subset M$ containing $\xi$. The {\it domain} of  $\mathcal{D}$ is the set ${\rm Dom}(\mathcal{D})$ of its regular points. If ${\rm Dom}(\mathcal{D})=M$, then $\mathcal{D}$ is called  {\it regular}. If a vector field $X$ takes values in $\mathcal{D}$, it is written $X\in \mathcal{D}$. We write $\mathfrak{X}(M)$ for the space of vector fields on $M$.

\begin{definition} Let $V$ be a Vessiot--Guldberg Lie algebra on $M$. 
	The so-called {\it distribution $\mathcal{D}^V$ associated with  $V$} takes the form
	$$
	\mathcal{D}^V_\xi:=\{X_\xi:X\in V\}\subset TM,\quad \forall \xi\in M.
	$$  
	A	{\it generic point} for $V$ is a regular point of  $\mathcal{D}^V$. The {\it domain} of $V$ is the set, ${\rm Dom}\,V$, of generic points of $V$.
\end{definition}

\begin{example}\label{DysReg} Consider the Lie algebra of vector fields  on $\mathbb{R}^2$ given by
	\begin{equation}\label{I4}
	{\rm I}_4:=\langle \partial_x+\partial_y, x\partial_x+y\partial_y, x^2\partial_x+y^2\partial_y \rangle.
	\end{equation} 
The rank of the distribution $\mathcal{D}^{{\rm I}_4}$ associated with I$_4$ at  $(x,y)\in \mathbb{R}^2$ is given by the rank of 
	\begin{equation*}
	M(x,y):=\begin{pmatrix}
	1 & x & x^2\\
	1 & y & y^2
	\end{pmatrix}.
	\end{equation*}
	The rank of $M(x,y)$ is two if and only if $x\neq y$. Hence, ${\rm Dom}(V)=\{(x,y)\in\mathbb{R}^2:x\neq y\}$.
\end{example}
\begin{definition} An {\it invariant distribution} of a Lie algebra $V$ of vector fields on $M$ is a distribution $\mathcal{D}$ on $M$ different from $M\times\{0\}$ and $TM$ satisfying that for every vector field $Y\in \mathcal{D}$ and $X\in V$, the vector field $[Y,X]$ takes values in $\mathcal{D}$. 
\end{definition}	

It is straightforward to see that a distribution $\mathcal{D}$ is invariant relative to $V$ if and only if the Lie bracket of any element from a fixed basis of $V$ and any vector field of a fixed family of vector fields spanning $\mathcal{D}$ takes values in $\mathcal{D}$.

\begin{example}
	The Lie algebra I$_4$ on $\mathbb{R}^2$ admits two-invariant distributions $\mathcal{D}^x$ and $\mathcal{D}^y$ generated by $\partial_x$ or $\partial_y$, correspondingly.  Indeed, $\mathcal{D}^x$ is invariant relative to I$_4$ because the Lie bracket of a generating element of $\mathcal{D}^x$, e.g. $\partial_x$, and any element of the basis (\ref{I4}) of I$_4$ belongs to $\mathcal{D}^x$:
	$$
	[\partial_x,\partial_x+\partial_y]=0, \quad [\partial_x, x\partial_x+y\partial_y]=\partial_x, \quad [\partial_x,x^2\partial_x+y^2\partial_y]=2x\partial_x.
	$$
	Similarly, it can be proved that $\mathcal{D}^y$ is an invariant distribution relative to I$_4$. 
\end{example}

\begin{definition}
	A finite-dimensional Lie algebra of vector fields $V$ on $ \mathbb{R}^2$ is {\it imprimitive}, if it admits an invariant  distribution. A Lie algebra $V$ is {\it one-imprimitive} if it has only one invariant distribution, and it is {\it multiply imprimitive} if it admits more than one. If $V$ has not invariant distributions, then $V$ is called {\it primitive}. 
\end{definition}

Lie proved \cite{GKP92,HA75} that every Vessiot--Guldberg Lie algebra  on $\mathbb{R}^2$ is locally diffeomorphic around a generic point to one of the Lie algebras described in  Table \ref{table3}. 

\section{Conformal geometry and Lie algebras of vector fields on the plane}

This section surveys the fundamentals on conformal geometry and related Vessiot--Guldberg Lie algebras of conformal and Killing vector fields to be employed hereupon. 

\begin{definition} A {\it pseudo-Riemannian manifold} is a pair $(M,g)$, where $g$ is a symmetric non-degenerate two-covariant tensor field on $M$: the {\it pseudo-Riemannian metric} of $(M,g)$.
\end{definition}

To simplify the notation, $g$ will be called a {\it metric} and Einstein summation convention will be assumed henceforth. 
In coordinates $\{x^i\}$ a metric $g$ is written as $g=g_{ij}{\rm d}x^i\otimes{\rm d}x^j$. 

\begin{example} A relevant role is played subsequently by the metrics 
	  $g_E:={\rm d}x\otimes {\rm d}x + {\rm d}y\otimes {\rm d}y$ and $g_H:={\rm d}x\otimes {\rm d}y + {\rm d}y\otimes {\rm d}x$ respectively. Every flat metric on $\mathbb{R}^2$ can be mapped into one of them, up to a non-zero multiplicative constant, by an appropriate diffeomorphism.
\end{example}

\begin{definition} A vector field  $X$ on $M$ is {\it conformal} relative to the metric $g$ if $\mathcal{L}_Xg=f_Xg$ for a certain function  	$f_X\in C^\infty(M)$. The function $f_X$ is called the {\it potential} of $X$. A {\it Killing vector field} is a conformal vector field with $f_X=0$.
\end{definition}
\begin{example} The {\it Schwarzschild metric} \cite{Wa84} is a metric given by 
	\begin{equation*}
	g_{S}:=\bigg(1-\frac{2M}{r}\bigg){\rm d}t\otimes {\rm d}t-\bigg(1-\frac{2M}{r}\bigg)^{-1}{\rm d}r\otimes {\rm d}r - r^2\bigg({\rm d}\theta\otimes {\rm d}\theta+\sin^2\theta{\rm d}\varphi\otimes{\rm d}\varphi\bigg),\,M>0.
	\end{equation*}
	This metric appears in the description of black holes \cite{Wa84,SW72}. The vector fields $\partial_t,\ \partial_\varphi$ are Killing vector fields relative to the Schwarzschild metric, namely $\mathcal{L}_{\partial_t}g_{S}=\mathcal{L}_{\partial_\varphi}g_S=0$.
\end{example}

\begin{definition} Two metrics $g_1$ and $g_2$ on $M$ are {\it conformally equivalent} if 
	\begin{equation*}
	\exists \Omega>0, \Omega\in C^\infty(M),\quad g_1=\Omega^2g_2.
	\end{equation*}
\end{definition}
\begin{definition}
	A pseudo-Riemannian manifold $(M,g)$ is {\it conformally flat} if $g$ is locally conformally equivalent to a flat metric, i.e. there exists for each $x\in M$ an open $U^x\ni x$ and a function $f\in C^{\infty}(U^x)$ such that $g=e^{2f}g_f$ on $U^x$ for a flat metric $g_f$ on $U^x$.
\end{definition}

The following well-known result will be of key importance in this work (see \cite{GLN13}).

\begin{theorem}\label{conf}
	Every metric on the plane is conformally flat.
\end{theorem}

Let us now discuss the Lie algebras of conformal and Killing vector fields relative to a flat  metric on $\mathbb{R}^2$. It follows from the definition of conformal and Killing vector fields that conformal vector fields relative to a metric $g$ on $M$ generate a Lie algebra containing, as a Lie subalgebra, the Killing vector fields relative to $g$. 

We now prove the following result, which ensures that the classification of Lie algebras of conformal vector fields on $\mathbb{R}^2$ relative to metrics can be reduced to the classification of Lie algebras of conformal vector fields relative  to $g_E$ and $g_H$.

\begin{proposition} All Lie algebras of conformal vector fields on $\mathbb{R}^2$ relative to definite (resp. indefinite) metrics are diffeomorphic.
	\end{proposition}
\begin{proof}
If $V$ is a Lie algebra of conformal vector fields relative to $g$, then $V$ is also a Lie algebra of conformal vector fields relative to any other conformally equivalent metric. This amounts to the fact that every conformal vector field relative to a metric is a conformal vector field relative to any conformally equivalent metric. Let us prove this. Let $X$ be a conformal vector field relative to a metric $g_1$ on $\mathbb{R}^2$ and let $g_2$ be a metric conformally equivalent to $g_1$. Hence, locally around each point of $\mathbb{R}^2$, it is possible to write $g_1=e^fg_2$ for a certain  function $f\in C^\infty(\mathbb{R}^2)$. Thus,
$$
f_Xg_1=\mathcal{L}_X{g_1}=\mathcal{L}_X{e^fg_2}=(Xe^f)g_2+e^f\mathcal{L}_Xg_2\Rightarrow \mathcal{L}_Xg_2=(f_X-Xf)g_2,
$$
and $X$ is a conformal vector field relative to $g_2$. 

Since all metrics on the plane are conformally flat, the Lie algebra of conformal vector fields of a general metric on $\mathbb{R}^2$ is the Lie algebra of conformal vector fields of a flat metric. Moreover, flat metrics can be mapped into $g_E$ and $g_H$ through a local diffeomorphism. Therefore, the Lie algebra of conformal vector fields relative to a flat metric is, up to a diffeomorphism, the Lie algebra of conformal vector fields with respect to $g_E$ or $g_H$ depending on whether the initial metric was definite or indefinite, respectively. 

In consequence, the Lie algebra of conformal vector fields relative to a metric on $\mathbb{R}^2$ is diffeomorphic to the Lie algebra of conformal vector fields relative to $g_E$, if the metric is definite, and to $g_H$, if the metric is indefinite. 
\end{proof}

The above proposition allows us to slightly generalize the typical definition of conformal Lie algebras on $\mathbb{R}^2$ in terms of flat metrics as follows.

\begin{definition} We call $\mathfrak{conf}(p,q)$ the abstract Lie algebra isomorphic to the Lie algebra of conformal vector fields relative to a metric on $\mathbb{R}^2$ with signature $(p,q)$.
\end{definition}

Let us now analyse the Lie algebras of conformal vector fields relative to a definite and indefinite metric on $\mathbb{R}^2$, which amounts to studying the Lie algebras of conformal vector fields relative to $g_E$ or $g_H$, respectively. 

Consider the flat metric $g_E$ on $\mathbb{R}^2$ and let $X=X^x\partial_x+X^y\partial_y$. If  $X$ is a conformal vector field relative to  $g_E$ on $\mathbb{R}^2$, then $\mathcal{L}_Xg_E=f_Xg_E$ for 
 a certain $f_X\in C^\infty(\mathbb{R}^2)$. Hence
$$\mathcal{L}_Xg_E
= 2\partial_xX^x{\rm d}x\otimes {\rm d}x + 2\partial_yX^y{\rm d}y\otimes {\rm d}y+\big(\partial_yX^x+\partial_xX^y\big)({\rm d}x\otimes {\rm d}y+{\rm d}y\otimes{\rm d}x)
=f_Xg_E,   
$$
which amounts to 
$
\partial_xX^y+\partial_yX^x=0, \partial_xX^x=\partial_yX^y=f_X/2.
$
 Therefore, $X$ is a conformal vector field relative to $g_E$ if and only if the complex function $f:\mathbb{C} \ni z:=x+{\rm i}y \mapsto X^x(x,y)+{\rm i}X^y(x,y)\in \mathbb{C}$, where $x,y\in \mathbb{R}$, satisfies the {\it Cauchy--Riemann conditions}. This implies that 
$$
\mathfrak{conf}(2,0)=\mathfrak{conf}(0,2)\simeq \big\{f\partial_z|\ f:\mathbb{C}\rightarrow \mathbb{C}\ {\rm is\ holomorphic}\big\},
$$
and the Lie algebra $\mathfrak{conf}(2,0)$, which is a realification of the referred to as {\it Witt algebra}, is infinite-dimensional. 

Let us now consider the Lie algebra of conformal vector fields  of the hyperbolic metric $g_H$. If  $X$ is a conformal vector field relative to $g_H$, then
$$
\mathcal{L}_Xg_H
=\big(\partial_xX^x+\partial_yX^y\big)g_H+2\partial_yX^x{\rm d}y\otimes {\rm d}y+2\partial_xX^y{\rm d}x\otimes{\rm d}x=f_Xg_H,   
$$
which occurs if and only if 
$
\partial_yX^x=\partial_xX^y=0,  \partial_xX^x+\partial_yX^y=f_X.
$
This implies that every conformal vector field relative to $g_H$ takes the form
$$
X=X^x(x)\partial_x+X^y(y)\partial_y\Longrightarrow \mathfrak{conf}(1,1)\simeq \mathfrak{X}(\mathbb{R})\oplus \mathfrak{X}(\mathbb{R}).
$$

The study of previous results and its posterior use in this work demand the analysis of the {\it pseudo-orthogonal Lie algebras}. These are the matrix Lie algebras of the form
$$
\mathfrak{so}(p,q):=\big\{A\in\mathfrak{gl}(p+q):A^T\eta+\eta A=0,\ \eta :={\rm diag}(\underbrace{+\ldots +}_p\underbrace{-\ldots -}_q) \big\},
$$
where $\mathfrak{gl}(p+q)$ is the space of $(p+q)\times(p+q)$ matrices with real entries. Then, it is simple to see that $ {\rm P}_7\simeq \mathfrak{so}(3,1)$,
where P$_7$ is given in Table \ref{table3}, is a Lie algebra of conformal vector fields relative to $g_E$. Moreover, ${\rm I}_{11}\simeq \mathfrak{so}(2,2)$ is a Lie algebra of conformal vector fields with respect to $g_H$.
The Lie algebras P$_7$ and I$_{11}$ are some of the most relevant Lie algebras treated in this work.

 On the other hand, two conformally equivalent metrics may admit different Lie algebras of Killing vector fields. This will be  illustrated by the following proposition and our forthcoming classification of Lie algebras of Killing vector fields on $\mathbb{R}^2$.
 
 \begin{proposition} If $V$ is a Lie algebra of Killing vector fields relative to a metric $g$ on $M$ and $\mathcal{D}^V=TM$, then the scalar curvature $R$ of $g$ is constant.
 	\end{proposition}
 \begin{proof} The Killing vector fields of $g$ are symmetries of the scalar curvature $R$ thereof. Since the vector fields in $V$ span the whole distribution $TM$ and $R$ is a function, it follows that $R$ must be a first integral of every vector field on $M$ and consequently $R$ must be a constant.
 	\end{proof}

\section{Vessiot--Guldberg Lie algebras of conformal and Killing vector fields on $\mathbb{R}^2$}
The work \cite{GL16} accomplished a classification of Vessiot--Guldberg Lie algebras of vector fields relative to two types of flat metrics on the plane: the Euclidean (definite) and hyperbolic (indefinite) ones. That work did not  highlight that all metrics on $\mathbb{R}^2$ are conformally flat and, as noted in previous sections, that being a Lie algebra of conformal vector fields relative to a metric $g$ amounts to being a Lie algebra of conformal vector fields of a Euclidean or hyperbolic metric. Hence, it is obvious that Lemmas 7.1 and 7.2 as well as Propositions 7.4  and 7.5 in \cite{GL16}, which only apply to Euclidean and hyperbolic metrics, can be generalized as follows.

\begin{lemma}\label{con} There exist no conformal vector fields $X_1,X_2$ relative to a conformally flat Riemannian metric on $\mathbb{R}^n$ such that $n>1$  and $X_1\wedge X_2=0$. 
\end{lemma}

Lemma \ref{con} cannot be extended to $\mathbb{R}$: the vector fields $X_1:=\partial_u$ and $X_2:=u\partial_u$ are linearly independent and conformal relative to ${\rm d}u\otimes {\rm d}u$ on $\mathbb{R}$ whereas $X_1\wedge X_2=0$.

\begin{lemma}\label{one}
	Let $V$ be a Lie algebra of conformal vector fields relative to a metric $g$ on $\mathbb{R}^2$ and let  $\mathcal{D}$ be an invariant distribution relative to $V$. Therefore,
	\begin{enumerate}
		\item 
		the distribution $\mathcal{D}^\perp$ perpendicular to $\mathcal{D}$, i.e.
		$$
		\mathcal{D}^\perp_\xi:=\{X_\xi\in T_\xi M:g_\xi(X_\xi,\bar X_\xi)=0,\forall \bar X_\xi\in \mathcal{D}_\xi\},\quad \forall \xi\in M,
		$$
	    is invariant relative to $V$.
		\item The Lie algebra of conformal vector fields relative to an indefinite metric on $\mathbb{R}^2$ has, at least, two invariant distributions generated by commuting vector fields $Y_1, Y_2$. 
		\item A conformal vector field relative to an indefinite metric on $\mathbb{R}^2$ can be brought into the form $Z=f^1_ZY_1+f^2_ZY_2$, where $Y_1f^2_Z=Y_2f^1_Z=0$ for some $f^1_Z, f^2_Z\in C^{\infty}(\mathbb{R}^2)$.
	\end{enumerate}
\end{lemma}

As a consequence of the above lemma, if  $V$ is an imprimitive Lie algebra of conformal vector fields relative to a definite metric, then it also leaves invariant an additional perpendicular distribution. Hence, $V$ is primitive or multiply imprimitive.

\begin{proposition}\label{ConLie} 
	The Lie algebras ${\rm I}_1$, 
	${\rm P}_1,{\rm P}_2
	,{\rm P}_3,{\rm P}_4,{\rm P}_7,{\rm I}^{\alpha=1}_8,{\rm I}^{r=1}_{14}$ are the Vessiot--Guldberg Lie algebras of conformal vector fields relative to a  definite metric on 
	$\mathbb{R}^2$. They constitute, up to a diffeomorphism, the Lie subalgebras of {\rm P}$_7$.
\end{proposition}

\begin{proposition}\label{ConLie2}
	The Lie algebras ${\rm I}_1-{\rm I}_4,{\rm I}_6,{\rm I}^{\alpha=1}_8,{\rm I}_9-{\rm I}_{11}, {\rm I}_{14B}, {\rm I}_{15B}$ are the Vessiot--Guldberg Lie algebras of conformal vector fields relative to an indefinite metric on $\mathbb{R}^2$. They are, up to a diffeomorphism, the Lie subalgebras of  {\rm I}$_{11}$. 
\end{proposition}

Let us now classify Lie algebras of Killing vector fields on  $\mathbb{R}^2$ relative to a metric. Our findings are summarised in Table \ref{table3}.

\begin{theorem} \label{kill1}Let $X_1,X_2,Y$ be Killing vector fields relative to a metric  $g$ on $\mathbb{R}^2$ such that $X_1\wedge X_2\neq 0$ and $[X_1,X_2]=0$. Then:
	\begin{itemize}
		\item[1)] The functions $g(X_i,X_j),\ i,j=1,2$, are constant,
		\item[2)] If $Y\wedge X_i=0$ for a fixed $i\in\{1,2\}$, then  $Y$ and the $X_j$ are orthogonal or  commute. 
	\end{itemize}
\end{theorem}
\begin{proof} Let us prove $1)$. Since $X$ is a Killing vector field for $g$ and $[X_1,X_2]=0$ by assumption, it follows that
	$$
	\mathcal{L}_{X_i}g(X_j,X_k)=(\mathcal{L}_{X_i}g)(X_j,X_k)+g(\mathcal{L}_{X_i}X_j,X_k)+g(X_j,\mathcal{L}_{X_i}X_k)=0,\quad i,j,k=1,2.
	$$
	As $X_1\wedge X_2\neq 0$ and $X_1,X_2\in \mathfrak{X}(\mathbb{R}^2)$, the tangent vectors $X_1(\xi),X_2(\xi)$ span $T_\xi \mathbb{R}^2$ for every $\xi\in \mathbb{R}^2$ and $\mathcal{L}_{X_i}g(X_j,X_k)=0$ for every $i,j,k=1,2$. In consequence,  ${\rm d}[g(X_j,X_k)]=0$  and $g(X_j,X_k)$ is a constant for $j,k=1,2$.
	
	Let us now prove 2) by assuming $Y\wedge X_1=0$. Hence, $Y=fX_1$ for a function
	$f\in C^\infty(\mathbb{R}^2)$. Using the part a) of the present theorem, we obtain that $0=\mathcal{L}_{fX_1}g(X_i,X_i)$. Since $Y$ is a Killing vector field for $g$ and $[X_1,X_2]=0$, it additionally follows that 
	$$
	0=\mathcal{L}_{fX_1}g(X_j,X_j)=	g(\mathcal{L}_{fX_1}X_j,X_j)+g(X_i,\mathcal{L}_{fX_1}X_j)=-2(X_jf)g(X_1,X_j).
	$$
	Therefore,   $X_j$ and $Y$ are orthogonal, i.e. $g(X_j,Y)=0$, or $[X_j,Y]=(X_jf)X_1=0$.
\end{proof}

\begin{lemma}\label{KKK}
	If $V$ is a Lie algebra of Killing vector fields relative to a metric $g$ on $\mathbb{R}^2$ and there exist $X_1, X_2\in V$ such that $[X_1,X_2]=0$ and $X_1\wedge X_2\neq 0$, then $g=c_{ij}\theta^i\otimes \theta^j$, where the $c_{ij}$ are constant and $\theta^1,\theta^2$ are dual one-forms to $X_1,X_2$, i.e.
	$\theta^i(X_j)=\delta_j^i$, $i,j=1,2$.
\end{lemma}
\begin{proof} Assumption $X_1\wedge X_2$ implies that $\theta^1\wedge \theta^2\neq 0$ and any metric $g$ on $\mathbb{R}^2$ can be brought into the form 
	$$g=g_{11}\theta^1\otimes \theta^1+g_{12}(\theta^1\otimes \theta^2 + \theta^2\otimes \theta^1)+g_{22}\theta^2\otimes \theta^2,
	$$
	for certain functions $g_{ij}\in C^\infty(\mathbb{R}^2)$.
 Since $[X_i,X_j]=0$ by assumption, it turns out that
	$$
	(\mathcal{L}_{X_i}\theta^j)(X_k)=X_i[\theta^j(X_k)]-\theta^j([X_i,X_k])=0,\qquad i,j=1,2.
	$$ 
	Since $X_1\wedge X_2\neq 0$ also by assumption, the Lie derivative $\mathcal{L}_{X_i}\theta^j$ vanishes on an arbitrary vector field, i.e. $\mathcal{L}_{X_i}\theta^j=0$ for $i,j=1,2$.
	If  $X_i$ is a Killing vector field relative to $g$, then $\mathcal{L}_{X_i}g=0,\  i=1,2$. Due to this reason and as 
	$\mathcal{L}_{X_i}\theta^j=0$ for $i,j=1,2$, it follows that
	\begin{equation*}
	\mathcal{L}_{X_i}g=(\mathcal{L}_{X_i}g_{11})\theta^1\otimes \theta^1+(\mathcal{L}_{X_i}g_{12})(\theta^1\otimes \theta^2 + \theta^2\otimes \theta^1)+(\mathcal{L}_{X_i}g_{22})\theta^2\otimes \theta^2 =0,\quad i=1,2. \label{eq:kill}
	\end{equation*}
	The last equality holds if and only if $\mathcal{L}_{X_i}g_{11}=\mathcal{L}_{X_i}g_{12}=\mathcal{L}_{X_i}g_{22}=0$ for $i=1,2$. Since $X_1\wedge X_2\neq 0$, this means that the $g_{ij}$ are constant for $i,j=1,2$ .
\end{proof}

Among the Vessiot--Guldberg Lie algebras on the plane (see Table \ref{table3}), we aim to classify those Lie algebras $V$ consisting of Killing vector fields relative to a metric $g$,
namely, $\mathcal{L}_{X}g=0,\ \forall X\in V$. The following proposition is a consequence of Lemma \ref{KKK}.

\begin{proposition}\label{wn23p}
	The Lie algebra {\rm I}$_{14B}$ consists of Killing vector fields only relative to Euclidean and hyperbolic metrics on $\mathbb{R}^2$.
\end{proposition}
\begin{proof} If the vector fields of I$_{14B}:=\langle \partial_x,\partial_y\rangle$ are Killing vector fields relative to a metric $g$, then Lemma \ref{KKK} ensures that $g=c_{xx}{\rm d}x\otimes {\rm d}x+c_{xy}({\rm d}x\otimes {\rm d}y+{\rm d}y\otimes {\rm d}x)+c_{yy}{\rm d}y\otimes {\rm d} y$ for certain constants $c_{xx},c_{xy},c_{yy}$. Then, if I$_{14B}$ is a Lie algebra of Killing vector fields relative to a metric, then the metric must be flat. Additionally, the previous form of $g$ ensures that $\mathcal{L}_Yg=0$ for any $Y\in {\rm I}_{14B}$ and constants $c_{xx},c_{xy},c_{yy}$. It is enough then to choose $c_{xx},c_{xy},c_{yy}$  in such a way that $g$ is non-degenerate to see that I$_{14B}$ is a Lie algebra of Killing vector fields in respect of Euclidean and hyperbolic metrics. 
\end{proof}

\begin{proposition}\label{wn23}
	The Lie algebras on the plane given by 
	${\rm P}^{\alpha\neq 0}_1, {\rm P}_4-{\rm P}_8, {\rm I}_6-{\rm I}_{11}, {\rm I}_{16}-{\rm I}_{20}
	$
	do not consist of Killing vector fields relative to any metric on $\mathbb{R}^2$.		
\end{proposition}
\begin{proof} Propositions
	\ref{ConLie} and \ref{ConLie2} ensure that P$_5$, P$_6$, P$_8$, I$_7$, I$_8^{\alpha\neq 1}$, I$_{10}$ and I$_{16}$--I$_{20}$ are not Lie algebras of conformal vector fields. Hence, they cannot be Lie algebras of Killing vector fields. Let us then focus on the remaining Lie algebras stated in this proposition
	\begin{equation}\label{Rem}
	{\rm P}_1^{\alpha\neq 0},{\rm P}_4,{\rm P}_7,{\rm I}_{6},{\rm I}_8^{\alpha= 1},{\rm I}_9,{\rm I}_{11}.
	\end{equation}
	
	Let us proceed by reduction to the absurd, and we assume that previous Lie algebras consist of Killing vector fields relative to a metric on $\mathbb{R}^2$.	Apart from I$_7$, all previous Lie algebras satisfy the conditions given in  Lemma \ref{KKK} for 
	$X_1=\partial_x$ and $X_2=\partial_y$. The dual one-forms to $X_1,X_2$ read $\theta_1={\rm d}x$ and $\theta_2={\rm d}y$. Hence,
	$$
	g=c_{xx}{\rm d}x\otimes {\rm d}x+c_{xy}({\rm d}x\otimes {\rm d}y+{\rm d}y\otimes {\rm d}x)+c_{yy}{\rm d}y\otimes {\rm d}y
	$$
	for certain constants $c_{xx},c_{xy},c_{yy}$. 
	
	$\bullet$ {\it Lie algebra {\rm P}$^{\alpha\neq 0}_1$:}
	Let us take $X_3:=\alpha(x\partial_x + y\partial_y)  +  y\partial_x - x\partial_y\in$ P$_1$ where $\ \alpha> 0$. Since  $X_3$ is a Killing vector field relative to $g$, then 
	$$
	\mathcal{L}_{X_3}g=2(\alpha c_{xx}-c_{xy}){\rm d}x\otimes {\rm d}x+(c_{xx}+\alpha c_{xy}-c_{yy})({\rm d}x\otimes {\rm d}y+{\rm d}y\otimes {\rm d}x)+2(c_{xy}+\alpha c_{yy}){\rm d}y\otimes {\rm d}y=0 
	$$
	and therefore condition $\mathcal{L}_{X_3}g=0$ amounts to
	$$
	2(\alpha c_{xx}-c_{xy})=(c_{xx}+\alpha c_{xy}-c_{yy})=2(c_{xy}+\alpha c_{yy})=0 \Rightarrow \alpha^2c_{xx}(2+\alpha^2)=0.
	$$
	Since $\alpha\neq 0$ by assumption, $c_{xx}=c_{xy}=c_{yy}=0$  and $g=0$. This is a contradiction and P$_1^{\alpha\neq 0}$ does not consists of Killing vector fields for any $g$ on $\mathbb{R}^2$.

	$ \bullet$ {\it Lie algebra {\rm P}$_4$}:
	In this case we choose $X_3=x\partial_x+y\partial_y\in $P$_4$. Then,
	\begin{equation}\label{eq5}
	\mathcal{L}_{x\partial_x+y\partial_y}(c_{xx}{\rm d}x\otimes {\rm d}x+c_{xy}({\rm d}x\otimes {\rm d}y+{\rm d}y\otimes {\rm d}x)+c_{yy}{\rm d}y\otimes {\rm d}y)=2 g.
	\end{equation}
	If  $X_3$ is a Killing vector field, then $\mathcal{L}_{X_3}g=0$ and $g=0$. This is a contradiction and hence P$_4$ does not consist of Killing vector fields relative to any metric on the plane.

	$\bullet$ {\it Lie algebras} I$_6$, I$_9$, I$_{10}$: All these Lie algebras contain the vector field $X_3=x\partial_x$. As $X_3$ is a Killing vector field by assumption, $\mathcal{L}_{X_3}g=0$. From  (\ref{eq5}) it follows that $g=0$, which is a contradiction. Hence, none of the previous Lie algebras consists of Killing vector fields relative to any $g$ on $\mathbb{R}^2$.
	
	$\bullet$ {\it Lie algebra} I$_8$: Since  $X_3:=x\partial_x+y\partial_y\in$ I$_8^{\alpha=1}$ must be a Killing vector field relative to $g$, then (\ref{eq5}) shows that $g=0$. 
	
	$\bullet$ {\it Lie algebras {\rm P}$_6$ and {\rm P}$_7$}: Since P$_4$ does not consist of Killing vector fields relative to any metric on $\mathbb{R}^2$ and P$_4$ is a Lie subalgebra of P$_7$,P$_6$, the Lie algebras P$_6$ and P$_7$ cannot consist of Killing vector fields for any metric on $\mathbb{R}^2$ neither.
\end{proof}


\begin{corollary}\label{wn22}
	If  $V$ is a  Lie algebra of vector fields on $\mathbb{R}^2$ containing linearly independent
	$X_1,X_2, X_3$ such that $[X_1,X_2]=[X_2,X_3]=0, [X_1,X_3]\neq 0, X_2\wedge X_3\neq 0, X_1\wedge X_3=0$, then $V$ is not a Lie algebra of Killing vector fields related to any metric.
\end{corollary}
\begin{proof}
	Let us prove our claim by reduction to contradiction. Since $X_1\wedge X_3=0$, it exists a non-zero function 
	$f\in C^\infty(\mathbb{R}^2)$ such that, $X_3=f(\xi)X_1,\ \forall \xi\in\mathbb{R}^2$. As $X_3$ is a Killing vector field by assumption, it satisfies that $\mathcal{L}_{X_3}g=0$ relative to a metric $g$. Also from assumption $X_1\wedge X_2\neq 0$. Hence, using Lemma \ref{KKK}, we find that 
	\begin{equation*}
	0=\mathcal{L}_{fX_1}[g(X_1,X_1)]=-2(X_1f)g(X_1,X_1),\quad 		0=\mathcal{L}_{fX_1}[g(X_1,X_2)]=-(X_1f)g(X_1,X_2).
	\end{equation*}
	Since $[X_1,X_3]\neq 0$ by assumption, if follows that $X_1f\neq 0$ and $g(X_1,X_1)=g(X_1,X_2)=0$. From this result and as $X_2\wedge X_1\neq 0$, it turns out that $g$ is degenerate. This is a contradiction, which finishes the proof.
\end{proof}
\begin{proposition}\label{KilP} The Lie algebras {\rm P}$^{\alpha=0}_1$, {\rm P}$_{3}$, and {\rm I}$_4$ are Lie algebras of Killing vector fields relative to some metrics on the plane.
\end{proposition}
\begin{proof}
	In coordinates $x,y$ on $\mathbb{R}^2$, every metric on $\mathbb{R}^2$ reads
	\begin{equation}\label{eq:metryka} 
	g=g_{xx}{\rm d}x\otimes{\rm d}x+
	g_{xy}({\rm d}x\otimes{\rm d}y+{\rm d}y\otimes{\rm d}x)+g_{yy}{\rm d}y\otimes{\rm d}y, 
	\end{equation}
	for certain functions $g_{xx},g_{xy},g_{yy}\in C^\infty(\mathbb{R}^2)$. Let us analyse the possible values of $g$ making the Lie algebras mentioned in this proposition to consist of Killing vector fields.
	
	$\bullet$ {\it Lie algebra} P$_{1}^{\alpha=0}$: In this case, we aim to determine functions $g_{xx},g_{xy},g_{yy}\in C^\infty(\mathbb{R}^2)$ such that the vector fields of
	$${\rm P}_1^{\alpha=0}=\big\langle X_1:=\partial_x,\ X_2:=\partial_y,\ X_3:=y\partial_x-x\partial_y\big\rangle$$
	become Killing vector fields relative to $g$, i.e. $\mathcal{L}_{X_k}g=0$ for $k=1,2,3$. Imposing this condition for $k=1,2$, we obtain 
	$$
	g=c_{xx}{\rm d}x\otimes {\rm d}x+c_{xy}({\rm d}x\otimes {\rm d}y+{\rm d}y\otimes {\rm d}x)+c_{yy}{\rm d}y\otimes {\rm d}y,\quad c_{xx},c_{xy},c_{yy}\in \mathbb{R}.
	$$
	Meanwhile, the condition below follows from the case $k=3$:
	$$
	\mathcal{L}_{y\partial_x-x\partial_y}g=(c_{xx}-c_{yy})({\rm d}x\otimes {\rm d}y+{\rm d}y\otimes {\rm d}x)+2c_{xy}({\rm d}y\otimes {\rm d}y-{\rm d}x\otimes {\rm d}x)=0.
	$$
	The last equality is satisfied if and only if 
	$
	c_{xx}=c_{yy}, c_{xy}=0.
	$
	Hence,  the  Lie algebra P$_1^{\alpha=0}$ is a Lie algebra of Killing vector fields only relative to a Euclidean metric
	$$
	g=c_{xx}({\rm d}x\otimes {\rm d}x+{\rm d}y\otimes {\rm d}y),\quad c_{xx}\in \mathbb{R}\backslash\{0\}.
	$$
	
	$\bullet$ {\it Lie algebra} P$_{3}$: Let us determine functions $g_{xx},g_{xy},g_{yy}\in C^\infty(\mathbb{R}^2)$ such that
	$${\rm P}_3=\big\langle X_1:=y\partial_x-x\partial_y,\ X_2:=(1+x^2-y^2)\partial_x+2xy\partial_y,\ X_3:=2xy\partial_x+(1+y^2-x^2)\partial_y\big\rangle$$
	consists of Killing vector fields relative to a metric $g$, i.e. $\mathcal{L}_{X_k}g=0$ for $k=1,2,3$. This condition for $k=1$ takes the form

	\begin{equation*}
	\left\{
	\begin{aligned}
	yg_{xx,x}-xg_{xx,y}-2g_{xy}&= 0,\\ 
	2g_{xy}+yg_{yy,x}-xg_{yy,y}&=0,\\
	g_{xx}-g_{yy}+yg_{xy,x}-xg_{xy,y}&=0.
	\end{aligned}
	\right.\Longleftrightarrow
	\left\{
	\begin{aligned}
	y(g_{xx}+g_{yy})_{,x}-x(g_{xx}+g_{yy})_{,y}&= 0,\\ 
	y(g_{xx}-g_{yy})_{,x}-x(g_{xx}-g_{yy})_{,y}-4g_{xy}&=0,\\
	g_{xx}-g_{yy}+yg_{xy,x}-xg_{xy,y}&=0,
	\end{aligned}
	\right.
	\end{equation*}
where every subscript given by a coordinate after a comma determines a derivative in that coordinate.
	These equations are satisfied when $g_{xx}=g_{yy}=:f(x,y)$ and $f(x,y)$ is such that $y(g_{xx}+g_{yy})_{,x}-x(g_{xx}+g_{yy})_{,y}= 0$  and
	$g_{xy}=0$. Hence, $X_1$ is a Killing vector field for 
	$$g_f:=f(x,y)[{\rm d}x\otimes{\rm d}x+{\rm d}y\otimes{\rm d}y].$$
	It is now time to determine those $f\in C^\infty(\mathbb{R}^2)$ satisfying that $\mathcal{L}_{X_k}g_f=0,\ k=2,3$. Hence,
	\begin{equation}\label{eq:KillMetX23}
	\left\{
	\begin{aligned}
	\mathcal{L}_{X_2}g&=((1+x^2-y^2)f_{,x}+2xyf_{,y}+4x)({\rm d}x\otimes{\rm d}x+{\rm d}y\otimes{\rm d}y)=0,\\
	\mathcal{L}_{X_3}g&=(2xyf_{,x}+(1+y^2-x^2)f_{,y}+4y)({\rm d}x\otimes{\rm d}x+{\rm d}y\otimes{\rm d}y)=0.\\
	\end{aligned} \right. \newline
	\end{equation}
	The system \eqref{eq:KillMetX23} amounts to
	\begin{equation}\label{eq11}
	\left\{
	\begin{aligned}
	(1+x^2-y^2)f_{,x}+2xyf_{,y}+4x&=0\\
	2xyf_{,x}+(1+y^2-x^2)f_{,y}+4y&=0
	\end{aligned} 
	\right.\Rightarrow
	\begin{pmatrix}
	f_{,x}\\ 
	f_{,y}
	\end{pmatrix}
	=-\frac{4}{1+x^2+y^2}
	\begin{pmatrix}
	x\\
	y
	\end{pmatrix},
	\end{equation}
	One of the non-zero solutions to (\ref{eq11}), away of $(0,0)$, is $f(x,y)=-2\log{(1+x^2+y^2)}$. Hence,  P$_3$ consists of Killing vector fields relative to the Riemannian metric
	\begin{equation}\label{su2}
	g=-2\log{(1+x^2+y^2)}({\rm d}x\otimes{\rm d}x+{\rm d}y\otimes{\rm d}y).
	\end{equation}
Since the Lie algebra P$_3$ is primitive, it does not admit invariant distributions.  Lemma \ref{one} ensures that  it is not a Lie algebra of Killing vector fields relative to any indefinite metric on $\mathbb{R}^2$.
	
	$\bullet$ {\it Lie algebra} I$_{4}$: We now study the Lie algebra I$_4$. In this case, $g_{xx},g_{xy},g_{yy}\in C^\infty(\mathbb{R}^2)$ must be found so as to ensure that the elements of the Lie algebra  
	$${\rm I}_4=\big\langle X_1:=\partial_x+\partial_y,\ X_2:=x\partial_x+y\partial_y,\ X_3:=x^2\partial_x+y^2\partial_y\big\rangle$$
	will become Killing vector fields relative to the metric $g$, i.e. $\mathcal{L}_{X_k}g=0$ for $k=1,2,3$. If $k=1$, then
	$$
	\mathcal{L}_{X_1}g_{xx}=\mathcal{L}_{X_1}g_{yy}=\mathcal{L}_{X_1}g_{xy}=0.
	$$
	In coordinates $\xi_1:=x-y$ and $\xi_2:=x+y$, the previous conditions imply  that $ \mathcal{L}_{X_1}f=2\partial_{\xi_2}f=0$, and then $f=f(x-y)$. Hence, $g_{xx}=h_{xx}(x-y),h_{yy}=h_{yy}(x-y),g_{xy}=h_{xy}(x-y)$ for some functions $h_{xx},h_{yy},h_{xy}\in C^\infty(\mathbb{R})$. 
	If $k=2$, we obtain the conditions
	$$
	\mathcal{L}_{X_2}g_{xx}+2g_{xx}=\mathcal{L}_{X_2}g_{xy}+2g_{xy}=\mathcal{L}_{X_2}g_{yy}+2g_{yy}=0.
	$$
	Thus,
	$$
	(x-y)h'_{xx}+2h_{xx}=(x-y)h'_{xy}+2h_{xy}=(x-y)h'_{yy}+2h_{yy}=0.
	$$
	Since the solution to $(x-y)f'(x-y)+2f(x-y)=0$ is $f(x-y)=\lambda/(x-y)^2,\ \lambda\in\mathbb{R}$, the metric $g$ takes the form
	$$
	g=\frac{1}{(x-y)^2}[c_{xx}{{\rm d}x}\otimes {{\rm d}x}+c_{xy}({\rm d}x\otimes {\rm d}y+{\rm d}y\otimes {\rm d}x)+c_{yy}{\rm d}y\otimes {\rm d}y],
	$$
	for some constants $c_{xx},c_{xy},c_{yy}$. Imposing the last condition, i.e. $k=3$, we reach to
	$$
	\mathcal{L}_{X_3}g=\frac{2}{(x-y)}[c_{xx}{{\rm d}x}\otimes {{\rm d}x}+c_{yy}{\rm d}y\otimes {\rm d}y]=0.
	$$
	Hence,  $c_{xx}=c_{yy}=0$ and I$_4$ is a Lie algebra of Killing vector fields relative to an indefinite metric. 
\end{proof}	

It is worth noting that the previous proposition ensures that P$^{\alpha=0}_1$ consists of Killing vector fields only relative to a metric equal, up to a non-zero proportional constant, to $g_E$. The Lie algebra P$_3$ consists of Killing vector fields only with respect to Riemannian metrics, and I$_4$ is a Lie algebra of Killing vector fields relative to an indefinite metric on $\mathbb{R}^2$ taking, up to a non-zero proportional constant, the form $g=({\rm d}x\otimes {\rm d}y+{\rm d}y\otimes {\rm d}x)/(x-y)^2$.  Since all previous Lie algebras have associated distributions of rank two and the curvature tensor $R$ for each metric is invariant under Killing vector fields, it follows that $R$ is covariant invariant and the corresponding spaces are {\it locally Riemannian}. 

Proposition \ref{KilP} can be reinterpreted as the consequence of the existence of a certain type of quadratic  Casimir element for the Lie algebras P$_1^{\alpha=0}$, I$_4$, P$_3$. Let us explain this relevant fact in detail, which will also allow us to describe all Lie algebras of Killing vector fields on $\mathbb{R}^2$ relative to arbitrary metrics.

Let $\mathfrak{g}$ be an abstract Lie algebra and let $\phi:\mathfrak{g}\rightarrow \mathfrak{X}(M)$ be a Lie algebra morphism. It is known that the {\it universal enveloping Lie algebra}, $U(\mathfrak{g})$, of $\mathfrak{g}$ is isomorphic to the {\it symmetric tensor algebra}, $S(\mathfrak{g})$, of $\mathfrak{g}$. This allows us to  extend $\phi$ to a unique morphism of associative algebras $\Upsilon:U(\mathfrak{g})\simeq S(\mathfrak{g})\rightarrow S(M)$, where $S(M)$ is the space of symmetric tensor fields on $M$. The Lie algebra $\mathfrak{g}$ induces a Lie algebra representation $\rho_\mathfrak{g}:\mathfrak{g}\rightarrow {\rm End}(U(\mathfrak{g}))$ by extending the derivation ${\rm ad}_v:w\in \mathfrak{g}\mapsto [v,w]\in \mathfrak{g}$, with $v\in \mathfrak{g}$, to a derivation $[v,\cdot]_{U(\mathfrak{g})}$ on $U(\mathfrak{g})$. If $V:=\phi(\mathfrak{g})$, then there exists a second Lie algebra representation ${\rho}_V:X\in V\mapsto \mathcal{L}_X\in {\rm End}(S(M))$, where $\mathcal{L}_X$ stands for the Lie derivative of symmetric tensor fields on $M$ relative to the vector field $X$.
It is easy to check that 
\begin{equation*}
\Upsilon([v,C]_{U(\mathfrak{g})})=\mathcal{L}_{\Upsilon(v)}\Upsilon(C),\qquad \forall v\in \mathfrak{g},\forall C\in U(\mathfrak{g}).
\end{equation*}
As a consequence, if $C\in U(\mathfrak{g})$ is a {\it Casimir element} of $\mathfrak{g}$, namely $[v,C]_{U(\mathfrak{g})}=0$ for all $v\in \mathfrak{g}$, then $\mathcal{L}_X \Upsilon(C)=0$ for every $X\in \phi(\mathfrak{g})$. 

Particular types of symmetric tensor fields of the form $\Upsilon(C)$, where $C$ is a Casimir for $\mathfrak{sl}(2)$, have appeared previously in \cite{BBHLS15}, where they were called {\it Casimir tensor fields}. Following this terminology, we will hereafter call the $\Upsilon(C)$, for $C$ being a Casimir for a certain Lie algebra, {\it Casimir tensor fields}.

\begin{theorem}\label{main}
Let $V$ be a Vessiot--Guldberg Lie algebra of vector fields whose isomorphic abstract Lie algebra $\mathfrak{g}$ admits a quadratic Casimir element $C\in U(\mathfrak{g})$ such that $\Upsilon(C)$ is non-degenerate. Then, $V$ consists of Killing vector fields relative to $\Upsilon(C)^{-1}$.
\end{theorem}

\begin{proof} Since $C$ is a Casimir element for $\mathfrak{g}$, it follows that $\mathcal{L}_X\Upsilon(C)=0$ for every $X\in V$, i.e. $\Upsilon(C)$ is a symmetric tensor field on $M$ invariant relative to the vector fields of $V$. Let us assume that $G:=\Upsilon(C)=g^{\mu\nu}\partial_\mu\otimes\partial_\nu$ in local coordinates. The equality $\mathcal{L}_XG=0$ for every $X\in V$ amounts, for $X=X^\alpha\partial_\alpha$, to
	\begin{equation}
	\label{eqI}
	(\mathcal{L}_XG)^{\mu\nu}=X^\alpha\partial_\alpha g^{\mu\nu}-(\partial_\alpha X^\mu) g^{\alpha\nu}-(\partial_\alpha X^\nu) g^{\alpha\mu}=0.
	\end{equation}
By assumption, $G$ is non-degenerate, i.e. the matrix $g^{\mu\nu}$ has an inverse $g_{\mu\nu}$. Let $g:=g_{\mu\nu}dx^\mu\otimes dx^\nu$. The coordinates of the Lie derivative of $g$ relative to $X$ read
\begin{equation}
\label{eqp}
(\mathcal{L}_Xg)_{\mu\nu}=X^\alpha\partial_\alpha g_{\mu\nu}+(\partial_\mu X^\alpha) g_{\alpha\nu}+(\partial_\nu X^\alpha) g_{\alpha\mu}.
\end{equation}
Substituting the equality  $\partial_{\alpha}g_{\mu\nu}=-g_{\mu\pi}(\partial_{\alpha}g^{\pi\kappa})g_{\kappa\nu}$ into (\ref{eqp}) and using (\ref{eqI}), it turns out that $\mathcal{L}_Xg=0$ and $V$ becomes a Lie algebra of Killing vector fields relative to the metric $g$.
\end{proof}

\begin{example} Let us apply Theorem \ref{main} to show that Lie algebra P$^{\alpha=0}_1$ consists of Killing vector fields relative to a metric on $\mathbb{R}^2$. Let $\mathfrak{g}$ be a Lie algebra isomorphic to P$^{\alpha=0}_1$ with a basis $v_1,v_2,v_3$ satisfying the same commutation relations as the basis of vector fields $X_1,X_2,X_3$ for P$^{\alpha=0}_1$ given in Table \ref{table3}. This gives a Lie algebra morphism $\phi:\mathfrak{g}\rightarrow \mathfrak{X}(\mathbb{R}^2)$ mapping each $v_i$ into $X_i$. The Lie algebra $\mathfrak{g}$ admits a quadratic Casimir element $v_1\otimes v_1+v_2\otimes v_2$. If $\Upsilon:U(\mathfrak{g})\rightarrow S(\mathbb{R}^2)$ is the corresponding associative algebra morphism, then  $\Upsilon(v_1\otimes v_1+v_2\otimes v_2)=X_1\otimes X_1+X_2\otimes X_2=\partial_x\otimes \partial_x+\partial_y\otimes \partial_y.$
	Hence, this tensor field is non-degenerate and the inverse is
	$$
	g={\rm d}x\otimes {\rm d}x+{\rm d}y\otimes {\rm d}y.
	$$
	This is essentially the metric detailed in Proposition \ref{KilP}.
\end{example}

\begin{example} Let us now employ the Theorem \ref{main} to show that the Lie algebra P$_2$ consists of Killing vector fields relative to a metric on $\mathbb{R}^2$. In view of Table \ref{table3}, the Lie algebra P$_2$ is isomorphic to an abstract Lie algebra $\mathfrak{sl}(2)$. Choose a basis $v_1,v_2,v_3$ thereof satisfying the same commutation relations as the basis of vector fields $X_1,X_2,X_3$ for P$_2$ given in Table \ref{table3}. This gives a Lie algebra morphism $\phi:\mathfrak{sl}(2)\rightarrow \mathfrak{X}(\mathbb{R}^2)$ mapping each $v_i$ into $X_i$. The Lie algebra $\mathfrak{sl}(2)$ admits a quadratic Casimir element $C:=v_1\otimes v_3+v_3\otimes v_1-2v_2\otimes v_2$. If $\Upsilon:U(\mathfrak{sl}(2))\rightarrow S(\mathbb{R}^2)$ is the corresponding associative algebra morphism, then  $\Upsilon(C)=X_1\otimes X_3+X_3\otimes X_1-2X_2\otimes X_2=-2y^2(\partial_x\otimes \partial_x+\partial_y\otimes \partial_y).$
	Hence, this tensor field is non-degenerate and the inverse is
	$$
	g=\frac{-1}{2y^2}({\rm d}x\otimes {\rm d}x+{\rm d}y\otimes {\rm d}y).
	$$
	A straightforward computation shows that indeed $g$ is invariant under the elements of P$_2$. Since $X_1,X_2$ span a Lie algebra diffeomorphic to I$_{14A}$ (cf. \cite{BBHLS15}), it follows that this Lie algebra also consists of Killing vector fields relative to $g$. 
\end{example}

\section{Invariant distributions for Vessiot--Guldberg Lie algebras on $\mathbb{R}^2$}
Lie proved that the Lie algebras $\{{\rm P}_i\}_{i=1,\ldots,8}$ do not admit any invariant distribution \cite{GKP92}, while Lie algebras  $\{{\rm I}_i\}_{i=1,\ldots,20}$ do. The knowledge of these distributions for Vessiot--Guldberg Lie algebras on $\mathbb{R}^2$ is relevant to the characterization of Vessiot--Guldberg Lie algebras of conformal vector fields on $\mathbb{R}^2$ (cf. \cite{GL16}). Although these distributions were employed in \cite{GL16}, it was not  detailed there how to obtain them. As a consequence, this work aims to determine the invariant distributions of the Lie algebras $\{{\rm I}_i\}_{i=1,\ldots,20}$. This task is accomplished by means of the following lemma.

\begin{lemma}\label{Lem1}
	If a Vessiot--Guldberg Lie algebra $V$ on $\mathbb{R}^2$ admits two vector fields $X_1,X_2$ such that $[X_1,X_2]=0$ and $X_1\wedge X_2\neq 0$, then every invariant distribution $\mathcal{D}$ for $V$ is spanned by a linear combinations $\lambda_1X_1+\lambda_2X_2$,  with $\lambda_1,\lambda_2\in\mathbb{R}$.
\end{lemma}
\begin{proof} Since  $X_1\wedge X_2\neq 0$, the invariant distribution $\mathcal{D}$ for $V$ can be generated by means of a vector field of the form
	$X_2$ or $X_1+\mu X_2$ for a certain function 
	$\mu\in C^\infty(\mathbb{R}^2$). If $\mathcal{D}$ is generated by $X_2$, then the lemma follows. If  $\mathcal{D}$ 
	is generated by $X_1+\mu X_2$, then there exist functions $f_i\in C^\infty(\mathbb{R}^2)$, with $i=1,2$, such that
	\begin{equation*}
	[X_i, X_1+\mu X_2]=(X_i\mu)X_2 =f_i(X_1 +\mu X_2),\qquad i=1,2.
	\end{equation*}
	Since $X_1\wedge X_2\neq 0$, it follows that  $f_1=f_2=0$. Moreover, $X_i\mu=0$ for $i=1,2$ and $\mu={\rm const}$. 
	Therefore, $\mathcal{D}$ is generated by $\lambda_1X_1+\lambda_2X_2$ for certains $\lambda_1,\lambda_2\in\mathbb{R}$.
\end{proof}

\begin{theorem}\label{DysNiezm}
	If  a Vessiot--Guldberg Lie algebra $V$ on $\mathbb{R}^2$ contains two linearly independent vector fields $X_1, X_2$ such that $[X_1,X_2]=0$ and $X_1\wedge X_2= 0$, then every distribution  $\mathcal{D}$ invariant relative to $V$ is generated by $X_1$.

\end{theorem}
\begin{proof}  Since  $X_1$, $X_2$ are linearly independent vector fields of $V$ satisfying $X_1\wedge X_2=0$ and $[X_1,X_2]=0$ by assumption, then there exists $f\in C^\infty(\mathbb{R}^2)$ such that $X_2=fX_1$ and $X_1f=0$. 
	Let $X_3$ be a vector field satisfying $X_1\wedge X_3\neq 0$. As $\mathcal{D}$ is a one-dimensional distribution and $X_1\wedge X_3\neq 0$, it is therefore generated by  $X_3$ or $X_1+\mu X_3$ for a certain  $\mu \in C^\infty(\mathbb{R}^2)$. If $\mathcal{D}$ is generated by $X_3$, then $[fX_1,X_3]=f_3X_3$ for a certain $f_3\in C^\infty(\mathbb{R}^2)$ and $[fX_2,X_1]=0$ by assumption. Since $X_1\wedge X_3\neq 0$, it follows that $X_3f=0$ and $X_1f=0$. Then, $f$ is a constant and $X_2$ and $X_1$ are linearly independent, which is a contradiction and shows that $\mathcal{D}$ cannot be spanned by $X_3$. Let us assume that $\mathcal{D}$ is generated by $X_1+\mu X_3$. Since $\mathcal{D}$ is invariant relative to		$X_1,X_2$, there exist functions $f_1,f_2\in C^\infty(\mathbb{R}^2)$ such that
	\begin{align}
	[X_1,X_1+\mu X_3]&=(X_1\mu)X_3+\mu [X_1,X_3]=f_1(X_1+\mu X_3), \label{eq:123}\\
	[fX_1,X_1+\mu X_3]&=f(X_1\mu)X_3+f\mu [X_1,X_3]-\mu(X_3f)X_1=f_2(X_1+\mu X_3) \label{eq:1234}
	\end{align}
	Substituting \eqref{eq:123} in \eqref{eq:1234} and recalling that $X_1\wedge X_3\neq 0$, we obtain that
	$$
	ff_1(X_1+\mu X_3)-\mu(X_3f)X_1=f_2(X_1+\mu X_3)\Rightarrow (ff_1-\mu X_3f-f_2)X_1+(ff_1\mu-f_2\mu)X_3=0.
	$$
	Since $X_3\wedge X_1\neq 0$, then $\mu X_3f=0$. We have two options, $\mu=0$ or $\mu\neq 0$. Let us assume that $\mu\neq 0$. Then, $X_3f=0$ and, as $X_1\wedge X_3\neq 0$ and $X_1f=0$ which is a consequence of the assumption $[X_2,X_1]=0$, we obtain that $f$ is a constant, which goes against our assumption that $X_1,X_2$ are linearly independent. In consequence, $\mu=0$ and $\mathcal{D}$ is generated by $X_1$. 
\end{proof}

\begin{corollary}The Lie algebras ${\rm I}_{12},{\rm I}_{13},{\rm I}_{16}-{\rm I}_{20}$ and ${\rm I}_{14},{\rm I}_{15}$ for $r>1$
	admit only one invariant distribution generated by $\partial_y$.
\end{corollary}
\begin{proof} In view of Table \ref{table3}, the above mentioned Lie algebras contain  the vector fields $X_1:=\partial_y,X_2:=\eta_1(x)\partial_y$. By applying then Theorem \ref{DysNiezm},  we obtain that every invariant distribution is generated by  $X_1$.

\end{proof}

\begin{theorem}\label{OstLem}
	Let $V$ be a Lie algebra containing some vector fields  $X_1,X_2,X_3$ on $\mathbb{R}^2$ such that
	$X_1\wedge X_2\neq 0$, $[X_1,X_2]=0$. Let $\mathcal{D}$ be an invariant distribution on $\mathbb{R}^2$ relative to $V$. Hence:
	\begin{enumerate}
		\item[a)] If  $[X_1,X_3]=X_2$ and $[X_3,X_2]=0$, then $\mathcal{D}$ is spanned by $X_2$,
		\item[b)] If  $[X_1,X_3]=X_1,$ then $\mathcal{D}$ is generated by $X_1$ or $X_2$.
	\end{enumerate}
\end{theorem}
\begin{proof} From the assumptions of this theorem and  Lemma \ref{Lem1} follow that the distribution $\mathcal{D}$ has to be generated by a linear combination  with real coefficients of $X_1,X_2$.   
	
	Let us prove a). Since $\mathcal{D}$ is invariant relative to $X_3$ by assumption, there exists $f_1\in C^\infty(\mathbb{R}^2)$ and $c_1,c_2\in \mathbb{R}$ with $c_1^2+c_2^2\neq 0$  such that
	$$
	[X_3,c_1X_1+c_2X_2]=-c_1X_2=f_1(c_1X_1+c_2X_2)\Rightarrow (f_1c_2+c_1)X_2+f_1c_1X_1=0.
	$$
	As $X_1\wedge X_2\neq 0$, then  $c_1=0$ and $\mathcal{D}$ is generated by $X_2$. 
	
	We now turn to prove b). Since $\mathcal{D}$ is invariant relative to $X_3$, there exist $f_1\in C^\infty(\mathbb{R}^2)$ and $c_1,c_2\in \mathbb{R}$ with $c_1^2+c_2^2\neq 0$ such that
	$$
	[X_3,c_1X_1+c_2X_2]=-c_1X_1=f_1(c_1X_1+c_2X_2)\Rightarrow c_1(f_1+1)X_1+f_1c_2X_2=0.
	$$
	Hence, there exist two possibilities: $f_1=0$ and therefore $c_1=0$, which implies $\mathcal{D}$ is generated by $X_2$; or $f_1\neq 0$, which gives $c_2=0$ and $\mathcal{D}$ is generated by $X_1$. 
\end{proof}

\begin{corollary}The Lie algebras {\rm I}$_6$,  {\rm I}$_9$, {\rm I}$_{10}$, and {\rm I}$_{11}$ have only two invariant distributions spanned by $\partial_x$ and $\partial_y$. The Lie algebra {\rm I}$_7$ has only one invariant distribution  spanned by $X=\partial_y$. 
\end{corollary}
\begin{proof} 	The vector fields of I$_7=\langle X_1,X_2,X_3,X_4\rangle$, where $X_1,\ldots,X_4$ are given in Table \ref{table3}, are such that $X_1,X_2,X_3$ obey the conditions of the case b) of Theorem \ref{OstLem}. Hence, their invariant distributions are generated by  $X_1$ or $X_2$. A straightforward computation shows that the only invariant distribution is  $X_2=\partial_y$. 
	
	Similarly, it can be proved that the invariant distributions for I$_6$, I$_9$, I$_{10}$, and I$_{11}$ are generated by  $X_1$ or $X_2$, where these vector fields are those ones indicated in Table \ref{table3}. A simple calculation shows that each of these vector fields generate an invariant distribution for the mentioned Lie algebras.
\end{proof}
\section{Applications in Physics}\label{Zast}

This section illustrates the physical relevance of systems of differential equations whose dynamic can be determined by Vessiot--Guldberg Lie algebras of conformal and Killing vector fields on $\mathbb{R}^2$ relative to a certain metric $g$. The results of previous sections are employed to construct $g$ and to prove that Vessiot--Guldberg Lie algebras consisting of Killing vector fields relative to $g$ are also Lie algebras of Hamiltonian vector fields relative to the symplectic structure induced by $g$. This much improves results in \cite{BBHLS15}, where such structures were obtained by long and tedious calculations.

\subsection{Milne--Pinney equations}
The Milne--Pinney equations, known by their many applications in Physics \cite{LA08} and mathematical properties \cite{Ru16}, take the form
\begin{equation}\label{MP}
\frac{{\rm d}^2x}{ {\rm d} t^2}=-\omega^2(t)x+\frac{c}{x^3},
\end{equation}
where $\omega(t)$ is any function depending on $t$ and $c\in \mathbb{R}$. If we define $y:={\rm d}x/ {\rm d}t$, the above differential equation can be rewritten as 
\begin{equation}\label{eq:FirstLie}
\left\{
\begin{aligned}
\frac{{\rm d} x}{ {\rm d} t}&=y,\\
\frac{{\rm d} y}{ {\rm d} t}&=-\omega^2(t)x+\frac{c}{x^3}.
\end{aligned} \right.
\end{equation}
System \eqref{eq:FirstLie} describes the integral curves of the $t$-dependent vector field (cf. \cite{JdL11}) $\mathcal{X}:=X_3+\omega^2(t)X_1,$ with
\begin{equation}\label{FirstLieA}
X_1=-x\partial_y,\qquad X_2=\frac 12 \left(y\partial_y-x\partial_x\right),\qquad X_3=y\partial_x+\frac{c}{x^3}\partial_y.
\end{equation}
The vector fields $X_1,X_2,X_3$  form a basis of a Lie algebra $V_{\rm MP}$.  Let us study $V_{\rm MP}$.
The matrix of its Killing form, $\kappa$, in the basis $\mathcal{B}:=\{X_1,X_2,X_3\}$ takes the form
$$
[\kappa]_\mathcal{B}=\left(\begin{array}{ccc}
0&0&-4\\
0&2&0\\
-4&0&0
\end{array}\right).
$$
Hence, the Killing form is non-degenerate and indefinite. The Cartan criterium \cite{trautman} ensures that the Lie algebra $V_{\rm MP}$ is semi-simple. Geometrically, Table \ref{table3} shows that every three-dimensional semi-simple Lie algebra of vector fields on the plane is isomorphic to $\mathfrak{sl}(2)$ or to $\mathfrak{so}(3)$. Algebraically, every semi-simple three-dimensional Lie algebra only admits such two options (cf. \cite{SW14}). Since $V_{\rm MP}$ is indefinite, $V_{\rm MP}$ is isomorphic to $\mathfrak{sl}(2)$.

Consider the Lie algebra $\mathfrak{sl}(2)$ and a basis $\{v_1,v_2,v_3\}$ thereof satisfying the same commutation relations as $X_1,X_2,X_3$. This induced a Lie algebra morphism $\phi:\mathfrak{sl}(2)\rightarrow \mathfrak{X}(\mathbb{R}^2)$ mapping each $v_i$ onto $X_i$. This gives rise to an associative algebra morphism $\Upsilon:U(\mathfrak{sl}(2))\rightarrow S(\mathbb{R}^2)$. The Lie algebra $\mathfrak{sl}(2)$ admits a quadratic Casimir element 
$$
C:=v_1\otimes v_3+v_3\otimes v_1-2v_1\otimes v_1.
$$
Therefore
$$
G:=\Upsilon(C)=X_1\otimes X_3+X_3\otimes X_1-2X_2\otimes X_2.
$$
In view of the coordinate expression for $X_1,X_2,X_3$, it follows that  
$$
G=-\frac{x^2}{2}\partial_x\otimes\partial_x-\left(\frac{2c}{x^2}+\frac{y^2}2\right)\partial_y\otimes \partial_y-\frac 12xy(\partial_x\otimes \partial_y+\partial_y\otimes\partial_x)\Rightarrow \det G=c.
$$
Hence, the tensor field $G$ is non-degenerate for $c\neq 0$. Then, Theorem \ref{main} ensures that the Lie algebra $V_{\rm MP}$ consists of Killing vector fields relative to
$$
g:=G^{-1}=-\left(\frac{2}{x^2}+\frac{y^2}{2c}\right){\rm d}x\otimes {\rm d}x+\frac{xy}{2c}({\rm d}x\otimes {\rm d}y+{\rm d}y\otimes {\rm d}x)-\frac{x^2}{2c}{\rm d}y\otimes {\rm d}y.
$$
The associated symplectic structure is given by $\omega:=\star 1$, i.e.
$$
\omega=\sqrt{|c|}{\rm d}x\wedge {\rm d}y.
$$
The vector fields of $V_{\rm MP}$ become Hamiltonian relative to $\omega$. In this simple manner, it was possible to obtain a symplectic form turning the elements of $V_{\rm MP}$ into Hamiltonian vector fields algebraically. Meanwhile, this result had to be obtained by solving a system of PDEs or by guessing the form of $\omega$ in previous works \cite{BBHLS15,CLS12}.

\subsection{Schr\"odinger equation on $\mathbb{C}^2$}
Let $\mathcal{H}$ be an $n$-dimensional Hilbert space with a  scalar product $\langle\cdot,\cdot \rangle$, let $H(t)\subset{\rm End}(\mathcal{H})$ be a Hermitian Hamiltonian operator on $\mathcal{H}$ for every $t\in\mathbb{R}$, and let $\{\psi_i\}_{i\in \overline{1,n}}\in\mathcal{H}$ be an orthonormal basis of quantum states, i.e $\langle\psi_i|\psi_j\rangle=\delta_{ij}$, $i\in \overline{1,n}$. It is possible to define in
$\mathcal{H}_0:=\mathcal{H}\backslash\{0\}$ an equivalence relation
$$\psi_1 \sim \psi_2 \Leftrightarrow \exists \lambda\in\mathbb{C}\backslash\{0\}: \psi_1=\lambda\psi_2,$$ which gives rise to the complex projective space $\mathcal{PH}:=\mathcal{H}_0/ \sim$ as its space of equivalence classes. Since this is also the space of orbits of the free and proper multiplicative action of the Lie group $\mathbb{C}_0:=\mathbb{C}\backslash\{0\}$ on $\mathbb{C}^n_0:=\mathbb{C}^n\backslash\{0\}$, the space $\mathcal{PH}:=\mathcal{H}_0/ \sim$ becomes a manifold.

Let $\mathbb{C}^2_0\ni \psi\mapsto [\psi] \in\mathbb{C}{P}^1\simeq \mathbb{C}^2_0/\mathbb{C}_0,\ \psi:=(z_1,z_2)$ be the projection from $\mathbb{C}_0^2$ onto its projective space. A $t$-dependent Schr\"odinger equation on  $\mathcal{H}$ induced by a $t$-dependent Hamiltonian $H(t)$ takes the form
\begin{equation*}
\frac{{\rm d}\psi}{ {\rm d}t}=-{\rm i}H(t)\psi\Leftrightarrow 
\frac{{\rm d}}{ {\rm d}t}
\begin{pmatrix}
z_1\\
z_2
\end{pmatrix}=-{\rm i}H(t)
\begin{pmatrix}
z_1\\
z_2
\end{pmatrix}=-{\rm i}
\begin{pmatrix}
\lambda_1(t) & b(t)\\
\bar{b}(t) & \lambda_2(t) 
\end{pmatrix}
\begin{pmatrix}
z_1\\
z_2
\end{pmatrix},\,\, 
\end{equation*}
for $b(t):=b_1(t)+{\rm i}b_2(t)$, $\lambda_i,b_i\in\mathbb{R}.$
If $\mu:=z_1z_2^{-1},\ z_1\in \mathbb{C},\ z_2\in\mathbb{C}_0$, then 
$$
\frac{{\rm d} \mu}{{\rm d}t}={\rm i}[\bar{b}(t)\mu^2+(\lambda_2(t)-\lambda_1(t))\mu-b(t)].
$$ 

Making a change of variables $\mu=x+{\rm i}y,\ x,y\in\mathbb{R}$,  and gathering together the parts real and imaginary of the previous system in the new variables, 
we obtain
\begin{equation*}
\left\{
\begin{aligned}
\frac{ {\rm d}x}{{\rm d}t}&=b_2(t)(x^2-y^2+1)-(\lambda_2(t)-\lambda_1(t))y-2b_1(t)xy\\ 
\frac{{\rm d}y}{ {\rm d}t}&=b_1(t)(x^2-y^2-1)+(\lambda_2(t)-\lambda_1(t))x+2b_2(t)xy,
\end{aligned}
\right.
\end{equation*}

\noindent describing the integral curves of the $t$-dependent vector field on $\mathbb{C}P^1$ of the form
\begin{equation*}
X=b_1(t)X_1+b_2(t)X_2+(\lambda_2(t)-\lambda_1(t))X_3, 
\end{equation*}
$$-X_1:=2xy\partial_x+(1+y^2-x^2)\partial_y,\quad X_2:=(x^2-y^2+1)\partial_x+2xy\partial_y,\quad -X_3:=y\partial_x-x\partial_y.$$

The vector fields $X_i,\ i=\{0,1,2\}$, span a Lie algebra $V_Q\!\!=$P$_3$. 
The Killing form, $\kappa$, of P$_3$ in the basis $\mathcal{B}:=\{X_1,X_2,X_3\}$ reads
$$[\kappa]_\mathcal{B}=
\left(\begin{array}{ccc}
-8&0&0\\
0&-8&0\\
0&0&-2
\end{array}\right).
$$
This Killing form is non-degenerate and negative-definite. As it is a three-dimensional semi-simple Lie algebra and there are only two semi-simple three-dimensional Lie algebras $\mathfrak{sl}(2)$ and $\mathfrak{so}(3)$,  Lie algebra $V_Q=\{X_1,X_2,X_3\}$ must be isomorphic to $\mathfrak{so}(3)$. In view of the Table \ref{table3}, this Lie algebra must be diffeomorphic to P$_3$. 
	\begin{landscape}
	\begin{table}[ht!] {\tiny
			\noindent
			\caption{{\small GKO Classification of Vessiot--Guldberg Lie algebras on $\mathbb{R}^2$. 
					Functions $\xi_1(x),\ldots,\xi_r(x)$ are linearly independent,  $\eta_1(x),\ldots,\eta_r(x)$ form a base of solutions to a linear system of $r$ linear differential equations with constant coefficients. We write $\mathfrak{g}=\mathfrak{g}_1\ltimes \mathfrak{g}_2$ 
					to indicate that $\mathfrak{g}$ is the direct sum of  $\mathfrak{g}_1$ and $\mathfrak{g}_2$, where  $\mathfrak{g}_2$ is an ideal  $\mathfrak{g}$. The symbol $`+'$ in the column Kill. indicates that a Lie algebra consists of Killing vector fields relative to metric and $`-'$ is written otherwise. The column Conf. details when a Lie algebra consists of conformal vector fields relative to a definite metric, ($g_E$), or a indefinite metric ($g_H$). The symbol $`-'$ means that a Lie algebra does not consist of conformal vector fields relative to any metric.}}
			\label{table3}
			\medskip
			\noindent\hfill
			\begin{tabular}{ p{.7cm} p{2cm}    p{8cm} p{.8cm}p{1.8cm}cc}
				\hline
				&  &\\[-1.9ex]
				\#&Primitive & Basis $X_i$ &Dom $V$&Inv. distribution& Kill. & Conf.\\[+1.0ex]
				\hline
				&  &\\[-2.9ex]
				P$_1$&$A_\alpha\simeq \mathbb{R}\ltimes \mathbb{R}^2$ & $  { {\partial_x} ,    {\partial_y} ,   \alpha(x\partial_x\!+\!y\partial_y) \!+\! y\partial_x\!-\!x\partial_y},\quad \ \alpha\geq 0$&$\mathbb{R}^2$&$-$&$+(\alpha=0)$&$g_E$ \\[+1.0ex]
				P$_2$&$\mathfrak{sl}(2)$ & $ {\partial_x},   {x\partial_x \!+\! y\partial_y} ,   (x^2 \!-\! y^2)\partial_x \!+\! 2xy\partial_y$&$\mathbb{R}^2_{y\neq 0}$&$-$&$+$&$g_E$\\[+1.0ex]
				P$_3$&$\mathfrak{so}(3)$ &${     { y\partial_x \!-\! x\partial _y},     { (1 \!+\! x^2 \!-\! y^2)\partial_x \!+\! 2xy\partial_y} ,   2xy\partial_x \!+\! (1 \!+\! y^2 \!-\! x^2)\partial_y}$&$\mathbb{R}^2$&$-$&$+$&$g_E$\\[+1.0ex]
				P$_4$&$\mathbb{R}^2\ltimes\mathbb{R}^2$ &$  {\partial_x},   {\partial_y},  x\partial_x\!+\!y\partial_y,   y\partial_x\!-\!x\partial_y$&$\mathbb{R}^2$&$-$&$-$&$g_E$\\[+1.0ex]
				P$_5$&$\mathfrak{sl}(2 )\ltimes\mathbb{R}^2$ &${  {\partial_x},   {\partial_y},  x\partial_x\!-\!y\partial_y,  y\partial_x,  x\partial_y}$&$\mathbb{R}^2$&$-$&$-$&$-$\\[+1.0ex]
				P$_6$&$\mathfrak{gl}(2 )\ltimes\mathbb{R}^2$ &${ {\partial_x},    {\partial_y},   x\partial_x,   y\partial_x,   x\partial_y,   y\partial_y}$&$\mathbb{R}^2$&$-$&$-$&$-$\\[+1.0ex]
				P$_7$&$\mathfrak{so}(3,1)$ &${  {\partial_x},   {\partial_y},   x\partial_x\!+\! y\partial_y,   y\partial_x \!-\! x\partial_y,   (x^2 \!-\! y^2)\partial_x \!+\! 2xy\partial_y,  2xy\partial_x \!+\! (y^2\!-\!x^2)\partial_y}$  &$\mathbb{R}^2$&$-$&$-$&$g_E$\\[+1.0ex]
				P$_8$&$\mathfrak{sl}(3 )$ &${  {\partial_x},    {\partial_y},   x\partial_x,   y\partial_x,   x\partial_y,   y\partial_y,   x^2\partial_x\!+\!xy\partial_y,   xy\partial_x \!+\! y^2\partial_y}$&$\mathbb{R}^2$&$-$&$-$&$-$\\[+1.5ex]
				\hline
				&  &\\[-1.5ex]
				\#& One-imprimitive\!\! & Basis $X_i$ &Dom $V$&Inv. distribution& Kill. & Conf.\\[+1.0ex]
				\hline
				&  &\\[-1.9ex]
				I$_5$&$\mathfrak{sl}(2 )$ &${ {\partial_x},    {2x\partial_x\!+\!y\partial_y},   x ^2\partial_x \!+\! xy\partial_y}$&$\mathbb{R}^2_{y\neq 0}$&$\partial_y$&$-$&$-$\\[+1.0ex]
				I$_7$&$\mathfrak{gl}(2 )$ & ${  {\partial_x},   {y\partial_y} ,     x\partial_x,    x^2\partial_x\!+\! xy \partial_y}$&$\mathbb{R}^2_{y\neq 0}$ &$\partial_y$&$-$&$-$\\[+1.0ex]
				I$_{12}$&$\mathbb{R}^{r\!+\!1}$ &$ {\partial_y} ,   \xi_1(x)\partial_y, \ldots , \xi_r(x)\partial_y,\quad   r\geq 1$&$\mathbb{R}^2$ &$\partial_y$&$-$&$-$\\[+1.0ex]
				I$_{13}$&$\mathbb{R}\ltimes \mathbb{R}^{r\!+\!1}$ &$  {\partial_y} ,   y\partial_y,    \xi_1(x)\partial_y, \ldots , \xi_r(x)\partial_y,\quad   r\geq 1$ &$\mathbb{R}^2$&$\partial_y$&$-$&$-$\\[+1.0ex]
				I$_{14}$&$\mathbb{R}\ltimes \mathbb{R}^{r}$ & ${ {\partial_x},   {\eta_1(x)\partial_y} ,  {\eta_2(x)\partial_y},\ldots ,\eta_r(x)\partial_y},\quad (r> 1,\ r=1,\ \eta_1'(x)\neq\eta_1(x))$&$\mathbb{R}^2$&$\partial_y$&$-$&$-$\\[+1.0ex]
				I$_{15}$&$\mathbb{R}^2\ltimes \mathbb{R}^{r}$ &  ${ {\partial_x},    {y\partial_y} ,    {\eta_1(x)\partial_y},\ldots, \eta_r(x)\partial_y},\quad  (r> 1,\ r=1,\ \eta_1'(x)\neq\eta_1(x))$&$\mathbb{R}^2$&$\partial_y$&$-$&$-$\\[+1.0ex]
				I$_{16}$&$C_\alpha^r\!\simeq\! \mathfrak{h}_2\!\ltimes\!\mathbb{R}^{r\!+\!1}$ & ${  {\partial_x},    {\partial_y} ,   x\partial_x \!+\! \alpha y\partial_y,   x\partial_y, \ldots, x^r\partial_y},\quad   r\geq 1,\qquad \alpha\in\mathbb{R}$&$\mathbb{R}^2$&$\partial_y$&$-$&$-$\\[+1.0ex]
				I$_{17}$&$\mathbb{R}\ltimes(\mathbb{R}\ltimes \mathbb{R}^{r})$ &$  {\partial_x},    {\partial_y} ,   x\partial_x \!+\! (ry \!+\! x^r)\partial_y ,   x\partial_y, \ldots,  x^{r\!-\!1}\partial_y,\quad   r\geq 1$ &$\mathbb{R}^2$&$\partial_y$&$-$&$-$\\[+1.0ex]
				I$_{18}$&$(\mathfrak{h}_2\!\oplus\! \mathbb{R})\!\ltimes\! \mathbb{R}^{r\!+\!1}$ & $  {\partial_x},    {\partial_y},   x\partial_x,   x\partial_y,   y\partial_y,   x^2\partial_y, \ldots,x^r\partial_y,\quad r\geq 1$ &$\mathbb{R}^2$&$\partial_y$&$-$&$-$\\[+1.0ex]
				I$_{19}$&$\mathfrak{sl}(2 )\ltimes \mathbb{R}^{r\!+\!1}$ &  $  {\partial_x},    {\partial_y} ,   x\partial_y,    2x\partial _x \!+\! ry\partial_y,   x^2\partial_x \!+\! rxy\partial_y,   x^2\partial_y,\ldots, x^r\partial_y ,\quad   r\geq 1$&$\mathbb{R}^2$ &$\partial_y$&$-$&$-$\\[+1.0ex]
				I$_{20}$&$\mathfrak{gl}(2 )\ltimes \mathbb{R}^{r\!+\!1}$ &  $  {\partial_x},    {\partial_y} ,   x\partial_x,   x\partial_y,   y\partial_y,   x^2\partial_x \!+\! rxy\partial_y,   x^2\partial_y,\ldots, x^r\partial_y,\quad   r\geq 1$ &$\mathbb{R}^2$&$\partial_y$&$-$&$-$\\[1.5ex]
				\hline
				&  &\\[-1.5ex]
				\#& multiply imprimitive\!\! & Basis $X_i$ &Dom $V$&Inv. distribution& Kill. & Conf.\\[+1.0ex]
				\hline
				&  &\\[-1.9ex]
				I$_1$&$\mathbb{R}$ &$  {\partial_x} $ & $\mathbb{R}^2$&$\partial_y,\partial_x+h(y)\partial_y$&$-$&$g_E,g_H$\\[+1.0ex]
				I$_2$&$\mathfrak{h}_2$ & $  {\partial_x} ,  x\partial_x$& $\mathbb{R}^2$&$\partial_x,\partial_y$&$-$&$g_H$\\[+1.0ex]
				I$_3$&$\mathfrak{sl}(2 )$  &$  {\partial_x},  x\partial_x,  x^2\partial_x$& $\mathbb{R}^2$&$\partial_x,\partial_y$&$-$&$g_H$\\[+1.0ex]
				I$_4$&$\mathfrak{sl}(2 )$  & ${  {\partial_x \!+\! \partial_y},    {x\partial _x\!+\!y\partial_y},   x^2\partial_x \!+\! y^2\partial_y}$ &$\mathbb{R}^2_{x\neq y}$&$\partial_x,\partial_y$&$+$&$g_H$\\[+1.0ex]
				I$_6$&$\mathfrak{gl}(2 )$ & ${ {\partial_x},    {\partial_y},   x\partial_x,   x^2\partial_x}$&$\mathbb{R}^2$&$\partial_x,\partial_y$&$-$&$g_H$\\[+1.0ex]
				I$^{\alpha\neq 1}_8$&$B_{\alpha\neq 1}\simeq \mathbb{R}\ltimes\mathbb{R}^2$ &${  {\partial_x},    {\partial_y},   x\partial_x \!+\! \alpha y\partial_y},\quad  0<|\alpha|< 1$&$\mathbb{R}^2$&$\partial_x,\partial_y$&$-$&$-$\\[+1.0ex]
				I$^{\alpha=1}_8$&$B_1\simeq \mathbb{R}\ltimes\mathbb{R}^2$ &${  {\partial_x},    {\partial_y},   x\partial_x \!+\! y\partial_y}$&$\mathbb{R}^2$&$\lambda_x\partial_x+\lambda_y\partial_y$&$-$&$g_E,g_H$\\[+1.0ex]
				I$_9$&$\mathfrak{h}_2\oplus\mathfrak{h}_2$ &${ {\partial_x},    {\partial_y},   x\partial_x,  y\partial_y}$&$\mathbb{R}^2$&$\partial_x,\partial_y$&$-$&$g_H$\\[+1.0ex]
				I$_{10}$&$\mathfrak{sl}(2 )\oplus \mathfrak{h}_2$ & ${ {\partial_x},    {\partial_y} ,   x\partial_x,  y\partial_y,  x^2\partial_x }$&$\mathbb{R}^2$&$\partial_x,\partial_y$&$-$&$g_H$\\[+1.0ex]
				I$_{11}$&$\mathfrak{sl}(2 )\oplus\mathfrak{sl}(2 )$ &$  {\partial_x},    {\partial_y},   x\partial_x,   y\partial_y,   x^2\partial_x ,   y^2\partial_y $&$\mathbb{R}^2$&$\partial_x,\partial_y$&$-$&$g_H$\\[+1.0ex]
				I$_{14A}$&$\mathbb{R}\ltimes \mathbb{R}$ & $ {\partial_x},    { e^{cx}\partial_y},\quad  c\in \mathbb{R}\backslash0$&$\mathbb{R}^2$&$e^{cx}\partial_y,\partial_x+cy\partial_y$&$+$&$g_E$\\[+1.0ex]
				I$_{14B}$&$\mathbb{R}\ltimes \mathbb{R}$ & $ {\partial_x},    {\partial_y}$&$\mathbb{R}^2$&$\lambda_x\partial_x+\lambda_y\partial_y$&$+$&$g_E,g_H$\\[+1.0ex]
				I$_{15A}$&$\mathbb{R}^2\ltimes \mathbb{R}$ &  $ {\partial_x},    {y\partial_y} ,    {e^{cx}\partial_y},\quad  c \in \mathbb{R}\backslash0$&$\mathbb{R}^2$&$e^{cx}\partial_y,\partial_x+cy\partial_y$&$-$&$-$\\[+1.0ex]
				I$_{15B}$&$\mathbb{R}^2\ltimes \mathbb{R}$ &  $ {\partial_x},    {y\partial_y} ,    \partial_y$&$\mathbb{R}^2$&$\partial_x,\partial_y$&$-$&$g_H$\\[+1.0ex]
				\hline
			\end{tabular}
			\hfill}
	\end{table}
	
\end{landscape}

The vector fields $X_1,X_2,X_3$ are exactly those ones of P$_3$.  Let us consider a basis $v_1,v_2,v_3$ of $\mathfrak{so}(3)$ satisfying the same commutation relations. 
This gives rise to an associative algebra morphism $\Upsilon:U(\mathfrak{so}(3))\rightarrow S(\mathbb{R}^2)$
The Lie algebra P$_3$ admits a quadratic Casimir element
$$
C=v_1\otimes v_1+v_2\otimes v_2+4v_3\otimes v_3.
$$
Then,
$$
G_0:=\Upsilon(C)=X_1\otimes X_1+X_2\otimes X_2+4X_3\otimes X_3=(1+x^2+y^2)^2(\partial_x\otimes \partial_x+\partial_y\otimes \partial_y).
$$
This $G_0$ is non-degenerate and Theorem \ref{main} allows us to construct a Riemannian metric $g$ turning the elements of $V_Q$ into Killing vector fields relative to
$$
g=G_0^{-1}=\frac{{\rm d}x\otimes{\rm d}x+{\rm d}y\otimes {\rm d}y}{(1+x^2+y^2)^2}.
$$ 
The symplectic structure related to $g$ takes the form
$$
\omega=\frac{{\rm d}x\wedge {\rm d}y}{(1+x^2+y^2)^2}=\star 1.
$$
In virtue of Theorem \ref{main}, the Lie algebra $V_{Q}$ is a Lie algebra of Killing vector fields relative to $\omega$. As in the previous section, this symplectic form is obtained algorithmically. This is much simpler than obtaining $\omega$ by solving a system PDEs as it was accomplished previously in the literature \cite{BBHLS15}. 

\section{Acknowledgements}
J. de Lucas acknowledges partial support from the Polish National Science Centre project MAESTRO under the grant  DEC-2012/06/A/ST1/00256.


\begin{thebibliography}{HD}

	\bibitem{BBHLS15}
A. Ballesteros, A. Blasco, F.J. Herranz, J. de Lucas, C. Sard\'on,
{\it Lie--Hamilton systems on the plane: properties, classification and applications}, J. Differential Equations 258 (2015), 2873--2907.
	
	\bibitem{BHLS15}
	A. Blasco, F.J. Herranz, J. de Lucas, C. Sard\'on,
	{\it Lie--Hamilton systems on the plane: applications and superposition rules}, 
	J. Phys. A {48} (2015), 345202.
	
	\bibitem{BL00}
	F. Boniver, P.B. Lecomte, 
	{\it A remark about the Lie algebra of infinitesimal conformal transformations of the Euclidean space}, 
	Bull. London Math. Soc. {32} (2000), 263--266.
	
	\bibitem{BC12}
	D. Burde, M. Ceballos, {\it Abelian ideals of maximal dimension for solvable Lie algebras}, J. Lie Theory {22} (2012), 741--756.
	
    	\bibitem{Ru16}
	R. Campoamor-Stursberg,
	{\it A functional realization of $\mathfrak{sl}(3,\mathbb{R})$ providing minimal Vessiot–Guldberg–Lie algebras of nonlinear second-order ordinary differential equations as proper subalgebras},
    J. Math. Phys.  {57} (2016), 063508.
	
	\bibitem{Ru16II}
	R. Campoamor-Stursberg,
	{\it Low dimensional Vessiot-Guldberg-Lie algebras of second-order ordinary differential equations},
    Symmetry {8} (2016), 15.

	\bibitem{JdL11}
	J.F. Cari{\~n}ena, J. de Lucas, {\sl Lie systems: theory, generalisations and applications},  Dissertationes Math. {479} (2011), 1--162.	
	
	\bibitem{CLS12}
	J.F. Cari\~nena, J. de Lucas, C. Sard\'on,
	{\it Lie--Hamilton systems: theory and applications},
	Int. J. Geom. Methods Mod. Phys. 10 (2013), 1350047.

	
	\bibitem{DLO99}
	C. Duval, P. Lecomte, V. Ovsienko, 
	{\it Conformally equivariant quantization: existence and
		uniqueness}, Ann. Inst. Fourier (Grenoble) { 49} (1999), 1999--2029.
	
	\bibitem{HLT17}
	F.J. Herranz, J. de Lucas, M. Tobolski,
	{\it 	Lie--Hamilton systems on curved spaces: A geometrical approach},
	arXiv:1612.08901.
	
	
	\bibitem{LO99}
	P.B.A. Lecomte, V.Y. Ovsienko, 
	{\it Projectively equivariant symbol calculus}, 
	Lett. Math. Phys. {49}  (1999), 173--196.
	
	\bibitem{GKP92}
	{A. Gonz{\'a}lez-L{\'o}pez, N. Kamran, P.J. Olver},
	{\it Lie algebras of vector fields in the real plane},
	{ Proc. London Math. Soc.} {64} (1992), 339--368.
	
	
	
	\bibitem{GL16}
	A.M. Grundland, J. de Lucas, {\it A Lie systems approach to the Riccati hierarchy and partial differential equations}, J. Differential Equations 263 (2017), 299--337.
	
		\bibitem{GLN13}
	A.C. Guti\'errez-Pi\~neres, C.S. L\'opez-Monsalvo, F. Nettel,
	{\it Two-dimensional Einstein manifolds in geometrothermodynamics},
	Adv. High Energy Phys. {2013} (2013), 967618.
	
	
	\bibitem{HA75}R. Hermann, M. Ackerman, {\it Sophus Lie's 1880 transformation group paper}, Math. Sci. Press, Brookline, 1975.
	
	\bibitem{KHBK15}
	S. Khan, T. Hussain, A.H. Bokhari and G. Ali Khan,	
	{\it Conformal Killing vectors of plane symmetric four dimensional Lorentzian manifolds}, Eur. Phys. J. C {75} (2015), 523.
	

	
	\bibitem{LA08}
	P.G.L. Leach, K. Andriopoulos,
	{\it The {E}rmakov equation: a commentary},
	Appl. Anal. Discrete Math. {2} (2008), 146--157.
	
	\bibitem{Li80}
	S. Lie, 
	{\it Theorie der Transformationsgruppen I}, 
	Math. Ann. 16 (1880), 441--528.
 
 \bibitem{Olver}
	P.J. Olver,
	{\sl Applications of Lie groups to differential equations},
	Springer-Verlag, New-York, 1986.

\bibitem{Ta49}
	A.H. Taub, 
	{\it A characterization of conformally flat spaces}, 
	Bull. Amer. Math. Soc. { 55} (1949), 85--89.
	
	\bibitem{trautman}
	A. Trautman, {\sl Teoria grup}, Skrypt FUW, 2011(Polish). Available at {http://www.fuw.edu.pl/~amt/skr4.pdf}.
		
	\bibitem{Wa84}
	R.M. Wald, {\sl General relativity}, University of Chicago Press, Chicago, 1984. 
	
	\bibitem{SW72}
	S. Weinberg, {\sl Gravitation and Cosmology: principles and applications of the general theory of relativity}, John Wiley and Sons, New York, 1972.

	\bibitem{SW14}
	P. Winternitz, L. \v{S}nobl,
	{\sl Classification and identification of Lie algebras.}, American Mathematical Society, Providence, 2014.
	
	\bibitem{Yo83}
	S. Yorozu,
	{\it Affine and projective vector fields on complete non-compact Riemannian manifolds},
	Yokohahas Math. {31} (1983), 41--46.
	

\end{thebibliography}
\end{document}